\documentclass[journal]{IEEEtran}

\usepackage{cite}

\ifCLASSINFOpdf
\usepackage[pdftex]{graphicx}
  \graphicspath{{images/}}
\else
\fi

\usepackage{color}

\usepackage{amsmath}
\usepackage{mathtools}
\usepackage{amsfonts}
\usepackage{amssymb}
\usepackage{amsthm}
\usepackage{bbm}

\usepackage[font=footnotesize,skip=3pt]{caption}

\usepackage{algorithm}
\usepackage[noend]{algpseudocode}
\usepackage{array}
\usepackage{url}
\usepackage{lipsum}
\usepackage{acro}

\DeclareMathOperator{\Prob}{\mathbb{P}}
\DeclareMathOperator{\E}{\mathbb{E}}

\usepackage{tikz}

\usetikzlibrary{arrows, arrows.meta, shapes, shapes.multipart, trees, positioning, backgrounds, fit, matrix, petri, calc, automata}

\usepackage{ellipsis}

\usetikzlibrary{calc}

\usetikzlibrary{decorations.pathreplacing,decorations.markings,shapes.geometric}

\makeatletter
\def\BState{\State\hskip-\ALG@thistlm}
\makeatother
\algdef{SE}[DOWHILE]{Do}{doWhile}{\algorithmicdo}[1]{\algorithmicwhile\ #1}

\newtheorem{lemma}{Lemma}

\newtheorem{definition}{Definition}

\tikzset{radiation/.style={{decorate,decoration={expanding waves,angle=90,segment length=4pt}}}}

\tikzset{cross/.style={cross out, draw=black, minimum size=2*(#1-\pgflinewidth), inner sep=0pt, outer sep=0pt},
	cross/.default={4pt}}

\usepackage[absolute,showboxes]{textpos}

\TPMargin{5pt}

\newcommand{\copyrightstatement}{
	\begin{textblock}{0.84}(0.08,0.01)    
		\noindent
		\footnotesize
		\copyright 2018 IEEE. This is the authors' version of the article. Personal use of this material is permitted. Permission from IEEE must be obtained for all other uses, in any current or future media, including reprinting/republishing this material for advertising or promotional purposes, creating new collective works, for resale or redistribution to servers or lists, or reuse of any copyrighted component of this work in other works.
	\end{textblock}
}

\begin{document}
\copyrightstatement
\title{On the Resource Consumption of M2M Random Access: Efficiency and Pareto Optimality\vspace{0.2cm}}
\author{Mikhail Vilgelm, Sergio Rueda Li\~nares, and Wolfgang Kellerer
\thanks{The authors are with the Chair of Communication Networks, Technical University of Munich, Germany. The work has been supported in part by the German Research Foundation (DFG)
	grant KE1863/5-1.}
\thanks{Corresponding author: M. Vilgelm (mikhail.vilgelm@tum.de).}
\vspace{-0.4cm}} 

\maketitle

\begin{abstract}
The advent of Machine-to-Machine communication has sparked a new wave of interest to random access protocols, especially in application to LTE Random Access (RA). By analogy with classical slotted ALOHA, state-of-the-art models LTE RA as a multi-channel slotted ALOHA. In this letter, we direct the attention to the resource consumption of RA. We show that the consumption is a random variable dependent on the contention parameters. We consider two approaches to include the consumption into RA optimization: by defining resource efficiency and by the means of a bi-objective optimization, where resource consumption and throughput are the competing objectives. We then develop the algorithm to obtain Pareto-optimal RA configuration under resource constraint. We show that the algorithm achieves lower burst resolution delay and higher throughput than the state-of-the-art.
\end{abstract}

%
\IEEEpeerreviewmaketitle

\vspace{-0.3cm}
\section{Introduction}
\IEEEPARstart{3}{GPP} LTE networks have multiple benefits if used as a basis for novel 5G applications, in particular for Machine-to-Machine (M2M) communications. Existing infrastructure, wide coverage area, and mature standardization mechanisms make them a leading candidate for M2M roll-outs. However, inherently sporadic and low datarate communication patterns have been shown to create a major bottleneck in the LTE Random Access Channel (RACH)~\cite{TR37868}. This observation sparked a new wave of interest to Random Access (RA) research, a field of study known since 1970s. In particular, RA is susceptible to high delays and low throughput in the case of synchronized burst arrivals, common for many M2M applications. Modeling RACH as multichannel slotted ALOHA (MS-ALOHA)~\cite{wei2015}, where a single Physical RACH (PRACH) preamble is considered as a \textit{channel}, inspired researchers to revisit classical S-ALOHA techniques, e.g., Rivest~\cite{rivest1987network}, and adapt these to MS-ALOHA. By analogy with S-ALOHA, barring, back-off, and conflict resolution techniques are applied to maximize the throughput of the RACH procedure~\cite{laya2014random}. 

A largely neglected aspect is the efficiency of LTE RA with respect to \textit{consumed time-frequency resources}. {\color{black}In the state of the art, resource consumption is either omitted or defined only in terms of the number of allocated preambles~\cite{laya2014random,7812676}. However, unlike a standard MS-ALOHA, LTE RA is a 4-step handshake (see Fig.~\ref{fig:resources}a). Since the collision occurs at its third step (MSG3), additional Physical Uplink Shared CHannel (PUSCH) resources are spent on MSG3 per every activated preamble, collided or successful. Therefore, random outcome of a preamble contention, i.e., the split between activated and idle preambles, directly influences the resource consumption. This makes \textit{idle preambles more favorable than collided in terms of the resource consumption} (Fig.~\ref{fig:resources}b). Thus, the state-of-the-art definition of resource consumption is incomplete and has to be revisited.}

In summary, this letter is motivated by three observations. First, resource consumption as a metric is ignored in the state of the art on M2M RA. Second, resource consumption is dependent on the outcome of the contention, and, hence, on the contention parameters. Third, as we show in the course of the letter, state-of-the-art algorithms are not optimal for the resource constrained optimization of RA.

In this letter, following these observations, we define the resource consumption of LTE RA (Sec.~\ref{sec:efficiency}) and derive its dependency from the contention parameters (access probability and number of allocated preambles). We then present and compare two approaches to take into account resource consumption: based on resource efficiency and based on Pareto optimality (Sec.~\ref{sec:pareto}). The latter approach is then developed into Pareto Optimal Channel allocation -- Access barring (POCA) algorithm for optimizing RA performance under resource constraints (Sec.~\ref{sec:jointlyopt}).

\vspace{-0.25cm}
\section{System Model}
\label{sec:system_model}

We consider a burst arrival scenario, where a large number of User Equipments (UEs) send requests to a single Base Station (BS) within a short period of time~\cite{TR37868}. Requests are resolved by the MS-ALOHA protocol, with $M$ preambles available per contention round. We assume a collision channel without capture and with ternary per-preamble feedback at the end of every contention round: idle ($0$ UEs chose a given preamble), singleton ($1$ UE), and collision ($>1$). 

Now consider an arbitrary contention round. Prior to it, $n$ backlogged UE, accounting for both previously unsuccessful and newly activated UEs, are competing for $M$ preambles. Access Class Barring (ACB) is applied with access probability $p$ (i.e. geometric random back-off). Instantaneous performance in a contention round is characterized with two performance metrics: number of successful UEs $s$, and \textit{resource consumption} $r$. We define the expectation of $s$ as \textit{throughput} $S$, which is a function of $n$, $p$, and $M$:
\begin{eqnarray}
S(n; p, M) \triangleq \E[s].
\label{eqn:expfunction}
\end{eqnarray}

The respective single contention round problem is:
\begin{eqnarray}
\underset{p,M}{\text{maximize }} S(n; p, M), \label{eqn:optimalsuccess}
\end{eqnarray}

where the $p$ and $M$ are contention parameters and $n$ is the input. Typically, only $p$ is considered as a parameter to be adjusted. In general, $M$ can also be adjusted dynamically and communicated to the UEs via BS broadcast~\cite{7812676,7404058}. BS can control the number of PRACH allocations per frame, thus reducing or increasing the number of available preambles. $M$ might also be set lower than the PRACH allocation allows, as the remaining preambles might be used for contention-free RACH or for another QoS class~\cite{7812676}. Note that solving the problem~\eqref{eqn:optimalsuccess} requires the knowledge of $n$, which is typically not available. Instead, the estimation of $n$ is used~\cite{newCollision,7404058}. For simplicity, we do not make a distinction between $n$ and its estimate throughout the next sections, but we use the estimation for evaluations in Sec.~\ref{sec:evaluation}.
\begin{figure}
	\centering
	\includegraphics[width=1\linewidth]{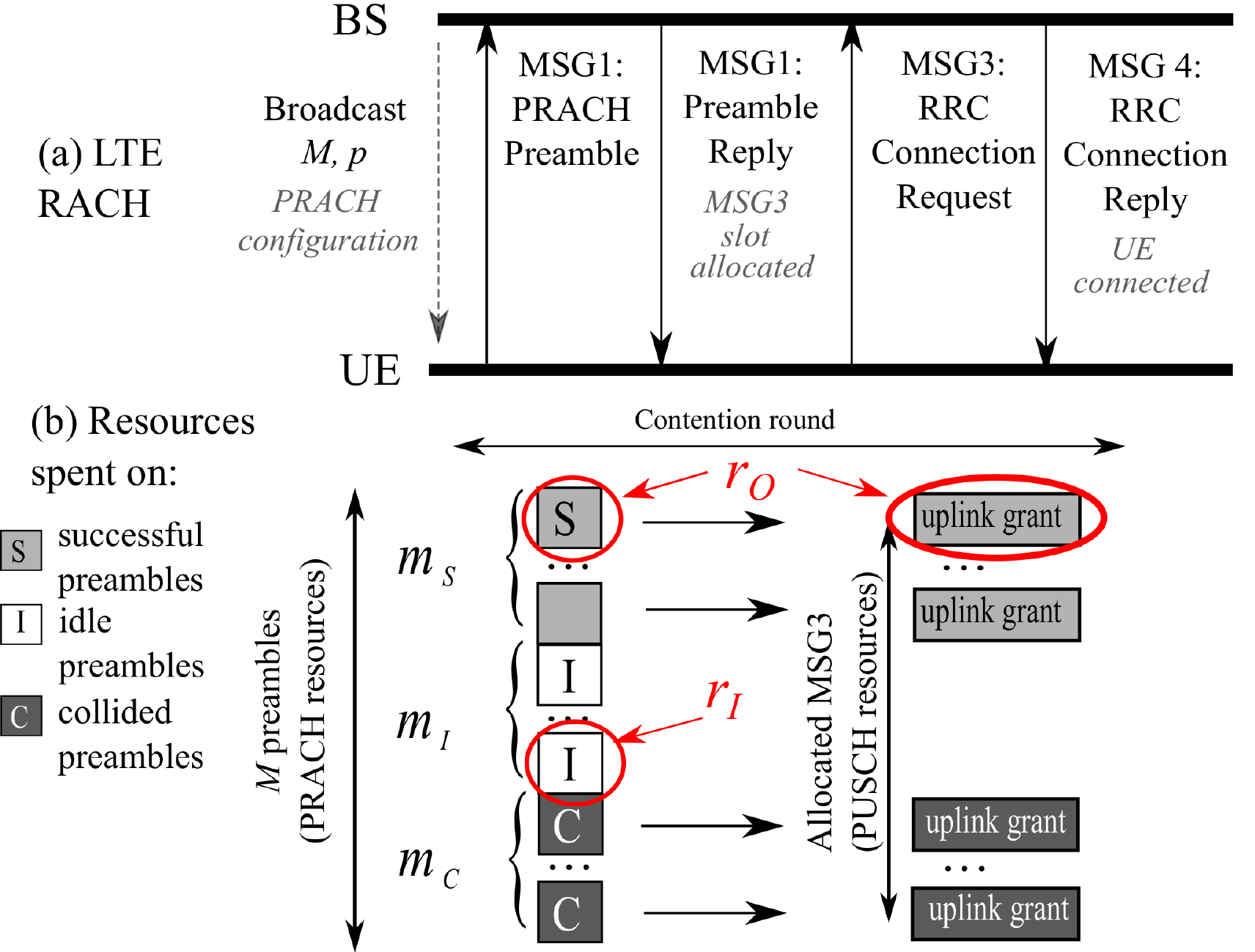}
	\caption{(a) Four step LTE RA; (b) Components of uplink resource consumption in LTE RA.}
	\label{fig:resources}
	\vspace{-0.5cm}
\end{figure}
\vspace{-0.3cm}
\subsection{Related Work}
Following an analysis similar to~\cite{wei2015}, it is straightforward to see that in the case of access barring without the re-transmission limit, the function~\eqref{eqn:expfunction} is expressed as follows:
\begin{equation}
S(n; p, M) = np\left(1-\frac{p}{M}\right)^{n-1}.
\label{eqn:exp-success}
\end{equation}
State-of-the-art papers search for the optimal access probability $p^\star$ and the number of channels $M^\star$ in various settings. In~\cite{7875393}, the authors propose a dynamic adaptation of access probability maximizing $S$
\begin{equation}
p^\star=\min(1,M/n)\label{eqn:optp}.
\end{equation}
A similar policy is adopted by Duan~\textit{et al.}~\cite{7404058}, with the additional step of adjusting the available preambles $M^\star$.

The state-of-the-art assumes that idle preambles are equal to collided in terms of the resource consumption and does not consider the resource constraint. In the next section, we define the resource consumption, and, based on it, redefine the RACH efficiency. Afterwards, we revisit the state-of-the-art policy~\eqref{eqn:optp}, and show that its performance is not optimal anymore if we account for the resource consumption\footnote{{\color{black}Due to space limitations, only most relevant papers are cited here. We refer the reader to~\cite{laya2014random,7812676} for extended review of RACH optimization.}}.
\vspace{-0.2cm}
\section{Efficiency and Pareto Optimality}
\label{sec:efficiency}

The total amount of uplink resources $r$ consumed by RACH in a contention round consists of two parts. First part are the resources consumed by PRACH preambles (MSG1). Second part are PUSCH resources: Every activated (occupied with $\geq 1$ UE) preamble is followed by an additional PUSCH allocation for MSG3. Since the number of occupied preambles is a random variable which depends on the contention parameters, consumed resources are also a random variable. 

To obtain $r$, we first define the vector $\mathbf{m}=[m_C,m_S,m_I]$, such that $M=m_C+m_S+m_I$, as an outcome of an arbitrary contention round, i.e. a ``split'' of $M$ available preambles into collided, successful, and idle (see Fig.~\ref{fig:resources}b). We define variables $r_I$ and $r_O$ as the amount of uplink resources spent per every idle or occupied preamble, respectively (assuming that collision and success consume equal resources). As illustrated in Fig.~\ref{fig:resources}b, $r_I$ includes only PRACH resources spent on 1 preamble, and $r_O$  includes PRACH+PUSCH resources\footnote{E.g., for LTE RACH: $r_I=6\text{ RBs}/64\text{ preambles}\approx 0.09$ RBs (only MSG1), $r_O=r_I+1\text{ RB}\approx 1.09$ RBs (MSG1+MSG3).}:
\begin{equation}
r \triangleq (m_S + m_C)r_O + m_Ir_I,
\end{equation}
which makes instantaneous consumption $r$ a random variable dependent on $\mathbf{m}$, with expectation $R(n; p, M)\triangleq \E[r]$.
\begin{definition}[Efficiency]
	The efficiency of RACH is defined as the ratio of the throughput to the expected resource consumption per contention round:
	\begin{eqnarray}
	T(n;p,M) &\triangleq& \frac{S}{R}.\label{eqn:efficiency}
	\end{eqnarray}	
\end{definition}
For the system model, the efficiency is found as:
\begin{eqnarray}
T(n;p,M) = \frac{S}{r_OM-\E[m_I](r_O-r_I)}.\label{eqn:efficiency_definition}
\end{eqnarray}
With $S$ given via~\eqref{eqn:exp-success}, it remains to find the expression for the expected number of idle preambles $\E[m_I]$. 

\begin{lemma}
	\label{lem:expnonidle}
	Given $n$ backlogged UEs, access probability $p$, the expected number of idle preambles $\E\left[m_I\right]$ in the contention round is:
	\begin{eqnarray}
	&&	\E\left[m_I\right] = M\left(1-\frac{p}{M}\right)^{n}\label{eqn:expectedidle}
	\end{eqnarray}
\end{lemma}
\begin{proof}
	{\color{black}Denote by $y_{j}$ the binary random variable indicating whether a preamble $j$ is idle in a given round. Any preamble is idle with probability 
		$	\Prob[y_{j}=1] = \left(1 - \frac{p}{M} \right)^{n}.$
		Using the sum of expectations $\E[m_I]=M\sum_{j}\E\left[y_{j}\right]=M\sum_{j}\Prob\left[y_{j}=1\right]$, we obtain~\eqref{eqn:expectedidle}.}
\end{proof}

Using Eqn.~\eqref{eqn:exp-success},~\eqref{eqn:efficiency_definition}, and the results of the lemma, we get
\begin{align}
&R=M\left(r_O-(r_O-r_I)\left(1-\frac{p}{M}\right)^{n}\right),\label{eqn:exp_consumption}\\
&T(n;p,M)= \frac{np\left(1-\frac{p}{M}\right)^{n-1}}{M\left(r_O-(r_O-r_I)\left(1-\frac{p}{M}\right)^{n}\right)}.
\label{eqn:throughputnew}
\end{align}

Our definition captures the difference in the resource consumption of idle and occupied channels, by weighting them differently in the expected outcome.

To illustrate the difference between throughput and efficiency, we plot the functions~\eqref{eqn:exp-success} and~\eqref{eqn:throughputnew} in Fig.~\ref{fig:norm_thr} against $p$ for different values of $M$. We observe that for the same value of $M$, different values of access probability $p^\star$ are needed to maximize $S$ and $T$. This result comes from the fact that the consumed resources $R$ are coupled with the channel split $\mathbf{m}$, making a collision less favorable, and, hence, reducing the access probability maximizing the efficiency. Consequently, the policy~\eqref{eqn:optp} is suboptimal in terms of efficiency $T$.

One approach to use this result would be to design an algorithm adjusting $(p,M)$ to maximize the efficiency. However, computing a jointly optimal solution is a non-linear mixed-integer problem and requires numerical methods. The worst-case complexity of such algorithm is not guaranteed to be polynomial. {\color{black}Additionally, since efficiency is a composite objective function, its usage presents \textit{a compromise between two competing metrics:} throughput and resource consumption. In the next sections, we present an alternative approach: To explore the contradicting nature of both metrics, we treat each of them as a separate objective in the framework of a bi-objective optimization problem.
	
	\begin{figure}[t!]
		\centering
		\includegraphics[trim={0 0.7cmm 0 0},clip,width=1\linewidth]{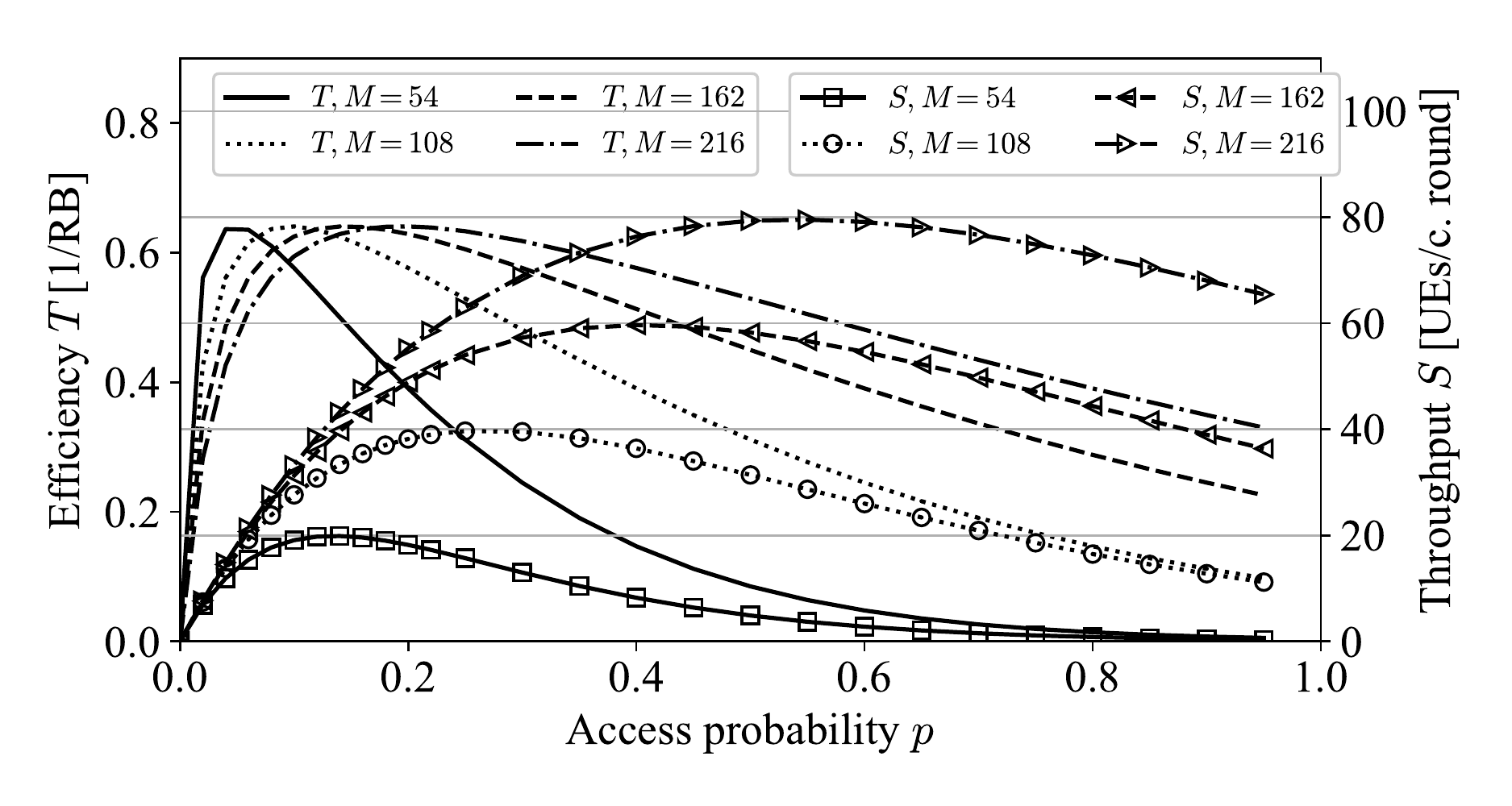}
		\caption{Resource efficiency $T$ and throughput $S$ vs. access probability $p$ for different $M$ values. $r_I=0.09$, $r_O=1.09$, $n_i=400$.}
		\label{fig:norm_thr}
		\vspace{-0.4cm}
	\end{figure}
	\vspace{-0.2cm}
	\subsection{RACH as Bi-Objective Optimization}
	\label{sec:pareto}
	We formulate the bi-objective optimization problem as follows:}
\begin{subequations}
	\label{eqn:multiobj}
	\begin{align}
	\underset{p,M}{\text{maximize }}&\left\lbrace S(n; p, M),-R(n; p, M)\right\rbrace,\label{eqn:pareto}\\
	\text{s.t. }& 
	p\in(0,1]\label{eqn:constr1}\\
	& M = km,\text{ }k\in\mathbb{N}_{++},\label{eqn:constr2}
	\end{align}
\end{subequations}
With~\eqref{eqn:constr2}, we impose an arbitrary constraint on the preamble allocation granularity: Preambles must be allocated as multiple integer of $m\geq 1$.

Now, we are looking for the Pareto set: values of $\left(p,M\right)$ for which none of the two objective functions can be increased without decreasing another objective. The sample solution space with the Pareto set is illustrated in Fig.~\ref{fig:constrained}. Every curve corresponds to a fixed value $M$ with varying $p$ for individual points. All the points on the lower border of the solution space form the Pareto set. We observe that the points $p^\star=\min\left(1,\frac{M}{n}\right)$, delivered by the state-of-the-art ACB policy~\eqref{eqn:optp}, do not belong to the Pareto set and hence are sub-optimal. The optimality gap can be read from Fig.~\ref{fig:constrained} as a respective projection of $p^\star$ points on the Pareto frontier.

{\color{black}The problem~\eqref{eqn:multiobj} can be solved by \textit{scalarization}~\cite{miettinen2008introduction}: Converting it into a single objective problem using the preferences between the objectives. We choose the $\epsilon$-constraint method: Set a constraint $\epsilon$ on one objective, and optimize for the second. A constraint can be set either on the minimum throughput or on the maximum resource consumption, depending on the practical use case and preferences of a generic decision maker. We take the second approach as an example and devise a practical algorithm to implement it in the next section.
	
	\begin{figure}
		\centering
		\includegraphics[trim={0 0.7cm 0 0},clip,width=1\linewidth]{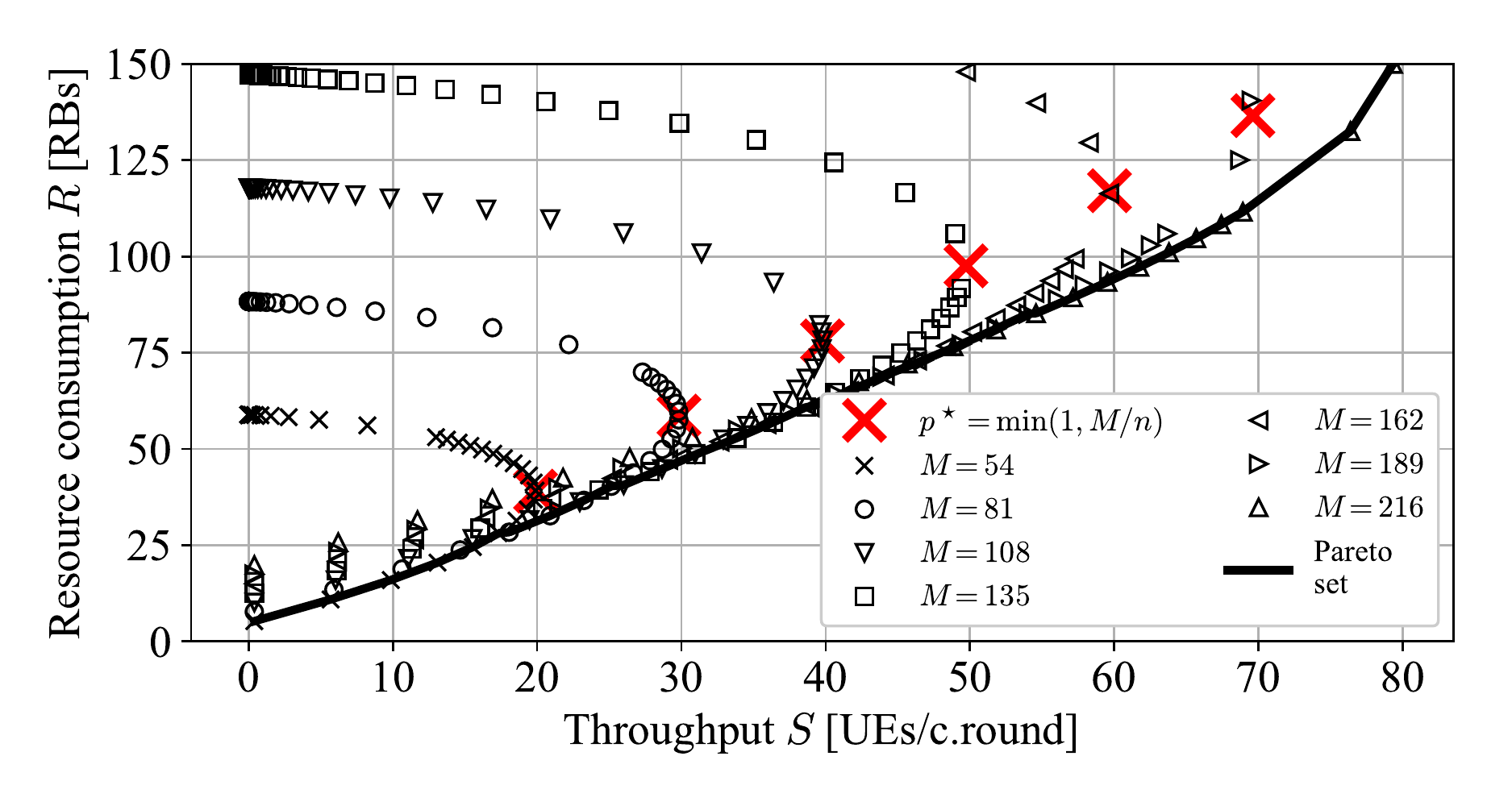}
		\caption{Functions $R$ and $S$ for varying $(p,M)$ as defined in~\eqref{eqn:multiobj} with Pareto set. ACB policy~\eqref{eqn:optp} is shown to deliver solutions for $p^\star$ diverging from the Pareto set. $r_I=0.09$, $r_O=1.09$, $n=400$ UEs.	}
		\label{fig:constrained}
		\vspace{-0.4cm}
	\end{figure}

\begin{figure*}[t!]
	\centering
	\includegraphics[trim={0 0.6cmm 0 0},clip,width=\linewidth]{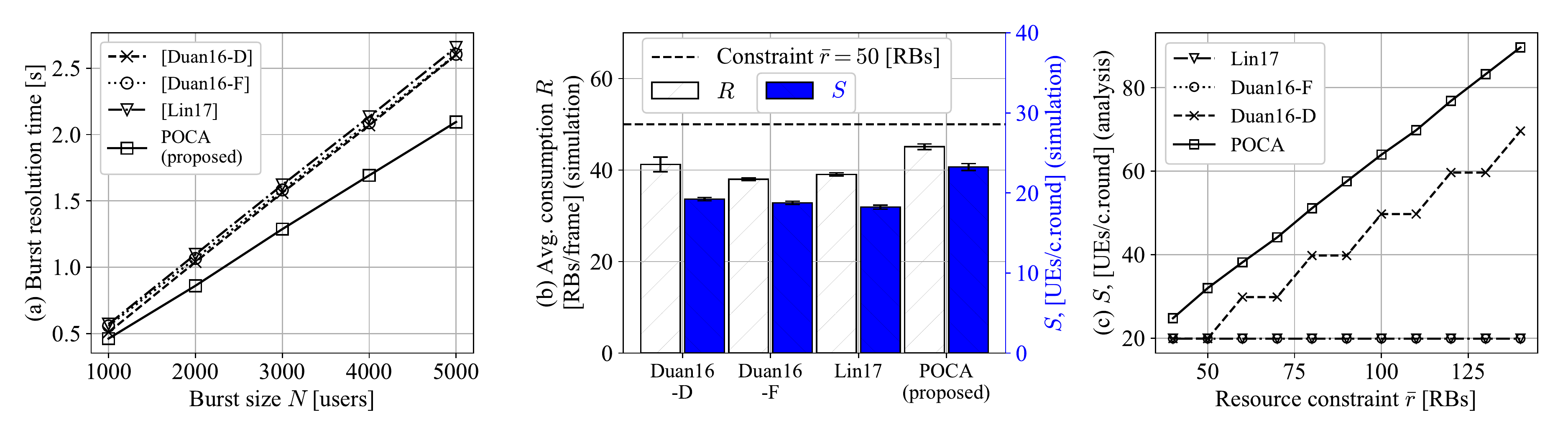}
	\caption{Evaluation results: (a) avg. burst resolution time (simulated); (b) avg. resource consumption $R$ and avg. throughput $S$ (simulated with burst size $N=2000$ UEs); (c) Throughput vs. constraint (analysis). Activation time $10$~ms, $\bar{r}=50$ RBs, $r_O=1.09$ RBs, $r_I=0.09$ RBs.}
	\vspace{-0.45cm}
	\label{fig:evaluation}
\end{figure*}

\vspace{-0.2cm}
\section{POCA: Pareto Optimal Channel allocation -- Access barring algorithm}
\label{sec:jointlyopt}

In this section, we design an algorithm to solve the problem~\eqref{eqn:multiobj} applying $\epsilon$-constraint method~\cite{miettinen2008introduction}. We treat RACH as a constrained optimization problem, maximizing the throughput $S$ given a certain constraint $\epsilon\equiv\bar{r}$ on the expected resource consumption\footnote{Separate constraints can also be enforced on PRACH and PUSCH resources. We argue however that a total constraint on the resources spent on connection establishment is more interesting. Both approaches can be accommodated in the framework with minor modifications.}. Since the channel split $\mathbf{m}$ is random, $\bar{r}$ is a constraint on expectation, which is a \textit{soft constraint}. It is a constraint which can be imposed by the desired dimensioning of resources in a system with dynamic scheduling (as in LTE). If, instead, desired is a hard constraint, it should be accommodated into $S$ and treated as unconstrained optimization problem instead~\cite{wei2015}.}

\vspace{-0.2cm}
\subsection{Constrained Optimization Problem}
By imposing the constraint $R\leq\bar{r}$, we reformulate~\eqref{eqn:optimalsuccess} as:\vspace{-0.15cm}
\begin{align}
\underset{p,M}{\text{maximize }}&S(n; p, M),\label{eqn:optimization}\\
\text{s.t. } & R\leq \bar{r},\text{ and~\eqref{eqn:constr1},~\eqref{eqn:constr2}}.\nonumber
\end{align}

The optimization problem is non-linear and mixed-integer, but a polynomial time solution can be found. First, note that for a fixed $M$, corresponding $p^\star$ has a closed form solution.

\begin{lemma}
For fixed $M$, the optimal solution $p^\star\left(M\right)$ to the problem~\eqref{eqn:optimization} is found as:
\begin{equation}
p^\star\left(M\right)=\min\left(\frac{M}{n},p_{\max}\right),\label{eqn:pstar}
\end{equation}
\begin{equation}
\text{where}\quad p_{\max}\triangleq \begin{cases}
M-M\left(\frac{r_O-\bar{r}/M}{r_O-r_I}\right)^{\frac{1}{n}} &\quad\text{if }r_O\geq \bar{r}/M,\\
1 &\quad\text{if }r_O<\bar{r}/M.\nonumber
\end{cases}
\end{equation}
\end{lemma}
\begin{proof}
Obtained by re-formulating the constraints.
\end{proof}

Second, the solution space with respect to $M$ is limited by the granularity $m$ and by $M_{\max}=\bar{r}/r_I$, obtained by setting $p=0$. Hence, the solution to the full problem reduces to the search in $k$, with computation of $p^\star$ for every iteration. Resulting complexity of the algorithm is hence at most linear $\mathcal{O}(k_{\max})$, where $k_{\max}=\lfloor\frac{\bar{r}}{r_Im}\rfloor$. The pseudocode for POCA is outlined in Alg.~\ref{alg:proposed_algorithm}.

\begin{algorithm}[b!]
{\footnotesize\caption{{\footnotesize POCA: Pareto Optimal Channel allocation -- Access barring}}\label{alg:proposed_algorithm}
	\begin{algorithmic}[1]
		\For {every contention round}
		\State Input: $n$, $\bar{r}$, $m$ (resource granularity). Set: $\hat{k}\gets k_{\max}$, $M^\star\gets\hat{k}m$; 
		\State Compute $p^\star=f(M^\star)$ via~\eqref{eqn:pstar}, $S^\star=f(n; p^\star,M^\star)$ via Eqn.~\eqref{eqn:exp-success}.
		\While {$\hat{k}>0$}
		\State Set: $\hat{k}\gets\hat{k}-1$, $\hat{M}\gets\hat{k}m$.
		\State Compute $\hat{p}= f(\hat{M})$ via~\eqref{eqn:pstar}, and $\hat{S}= f(n; \hat{p},\hat{M})$ via Eqn.~\eqref{eqn:exp-success}.
		\If {$\hat{S}>S^\star$}
		\State Set: $p^\star\gets\hat{p}$, $M^\star\gets\hat{M}$, $S^\star\gets\hat{S}$
		\EndIf
		\EndWhile 
		\Return $p^\star, M^\star$
		\EndFor
\end{algorithmic}}\vspace{-0.0cm}
\end{algorithm}
\vspace{-0.2cm}
\subsection{Evaluation and Benchmarking}
\label{sec:evaluation}
As benchmarks for POCA we choose preudo-bayesian approach (Lin17)~\cite{7875393}, dynamic access barring with fixed resource allocation (Duan16-F) and with dynamic resource allocation (Duan16-D)~\cite{7404058}. The first two, Lin17 and Duan16-F, do not adjust the number of preambles and only tweak the access probability, while the latter is optimizing both $p, M$. We aided the benchmarked algorithms with explicit constraint on the consumption $\hat{r}$. We set the constraint high enough for Duan16-F and Lin17 to deliver feasible solution. For Duan16-D, we find the best solution satisfying the constraint via exhaustive search. We have compared the performance in terms of burst resolution time, throughput and average resource consumption for a burst arrival scenario~\cite{TR37868}. All algorithms require knowledge of current backlog $n$, i.e., how many nodes will attempt the transmission in the next step. The estimation based on the observations of the channel split $\mathbf{m}$~\cite{7875393} is used here, as it performed best in our simulations. The simulation set-up follows the assumptions in~\ref{sec:system_model}, capturing only the MAC layer effects with parameters summarized in Fig.~\ref{fig:evaluation}.

In Fig.~\ref{fig:evaluation}a, average burst resolution time is plotted as a function of the burst size $N$. We observe that for $N=1000$ UEs the proposed algorithm achieves $9$\% lower time compare to the closest Duan16-D algorithm, and the gain grows to $19$\% for large bursts of $N=5000$ UEs. To study why is it the case, we plot the measured throughput and average resource consumption per contention round in Fig.~\ref{fig:evaluation}b for an exemplary burst size $N=2000$. We observe that POCA achieves higher throughput, and better utilizes the resources within the constraint $\bar{r}$. The deficiency of the benchmark solutions comes from the sub-optimality of the policy~\eqref{eqn:optp}, which we have demonstrated earlier. In contrast to it, POCA jointly considers channel allocation and access probability, hence, delivering the solutions from the Pareto set. We observe the difference in throughput as a function of the resource constraint in Fig.~\ref{fig:evaluation}c. Duan16-F and Lin17 do not adjust $M$, hence they achieve same throughput independent on the constraint. The throughput difference between Duan16-D and POCA is growing with the constraint, which is well inline with the observations in Fig.~\ref{fig:constrained}, where the optimality gap is increasing with resource consumption.
\vspace{-0.2cm}
\section{Conclusions}
In this letter, we have revisited the topic of the performance optimization of LTE RA. We have refined the definition of the resource consumption as a function of the contention parameters. We presented two approaches to accommodate the consumption in the optimization: by the means of efficiency or via the bi-objective optimization framework. We have shown that the state-of-the-art policies are suboptimal in terms of both efficiency and Pareto-optimality. We have presented an alternative algorithm Pareto-Optimal Channel allocation -- Access barring (POCA), and demonstrated that it is a superior approach for RA throughput optimization under resource constraint, achieving lower burst resolution delay.

\ifCLASSOPTIONcaptionsoff
  \newpage
\fi
\bibliographystyle{IEEEtran}

\end{document}